\documentclass{ieeetj}
\usepackage{cite}
\usepackage{amsmath,amssymb,amsfonts}
\usepackage[noend]{algorithmic}
\usepackage{graphicx,color}
\usepackage{textcomp}
\usepackage{xcolor}
\usepackage{hyperref}
\hypersetup{hidelinks=true}
\usepackage{algorithm, algorithmic}
\usepackage{todonotes}

\usepackage{tikz}

\makeatletter
\newcommand{\xleftrightarrow}[2][]{\ext@arrow 3359\leftrightarrowfill@{#1}{#2}}
\makeatother

\newcommand{\xdasharrow}[2][->]{
\tikz[baseline=-\the\dimexpr\fontdimen22\textfont2\relax]{
\node[anchor=south,font=\scriptsize, inner ysep=1.5pt,outer xsep=2.2pt](x){#2};
\draw[shorten <=3.4pt,shorten >=3.4pt,dashed,#1](x.south west)--(x.south east);
}
}

\def\BibTeX{{\rm B\kern-.05em{\sc i\kern-.025em b}\kern-.08em
    T\kern-.1667em\lower.7ex\hbox{E}\kern-.125emX}}
\AtBeginDocument{\definecolor{tmlcncolor}{cmyk}{0.93,0.59,0.15,0.02}\definecolor{NavyBlue}{RGB}{0,86,125}}

\usepackage{subfig}
\newtheorem{definition}{Definition}
\newtheorem{lemma}{Lemma}
\newtheorem{theorem}{Theorem}

\def\OJlogo{\vspace{-4pt}$<$Society logo(s) and publication title will appear here.$>$}
\def\seclogo{\vspace{10pt}$<$Society logo(s) and publication title will appear here.$>$}

\def\authorrefmark#1{\ensuremath{^{\textbf{#1}}}}

\begin{document}
\receiveddate{XX Month, XXXX}
\reviseddate{XX Month, XXXX}
\accepteddate{XX Month, XXXX}
\publisheddate{XX Month, XXXX}
\currentdate{XX Month, XXXX}
\doiinfo{XXXX.2022.1234567}

\markboth{}{Author {et al.}}

\title{Pattern-level Differential Privacy for High-utility Complex Event Processing}

\author{He Gu\authorrefmark{1}, Thomas Plagemann\authorrefmark{1}, Vera Goebel\authorrefmark{1}, Maik Benndorf \authorrefmark{1},\\ and Boris Koldehofe\authorrefmark{2}}

\affil{Department of Informatics, University of Oslo, Oslo, 0316 Norway}
\affil{Technische Universität Ilmenau, Ilmenau, Germany}
\corresp{Corresponding author: He Gu (emailco: heg@ifi.uio.no).}
\authornote{This work was funded by the Parrot Project (Research Council of Norway, project number 311197).}

\begin{abstract}
Current privacy-preserving mechanisms (PPMs) in Complex Event Processing (CEP) systems are unnecessarily restrictive, reducing the utility of data received by data consumers. This article presents a novel approach to preserve privacy in CEP systems, improving the utility of detected event patterns by dynamically adapting the noise added to an unprotected data stream. We introduce a new guarantee named pattern-level differential privacy (DP), which enables us to apply and compare the strength of PPMs at the pattern level. We propose new pattern-level PPMs yielding pattern-level DP and analyze different trust settings of these PPMs and their requirements for context knowledge in the CEP system, e.g., the deployed queries. Our evaluation is based on three datasets (two real-world, one synthetic) and shows that the proposed PPMs increase data utility while preserving the same privacy level as the state-of-the-art PPMs. We use simulations to study the performance of our proposed PPMs in various practical scenarios. Furthermore, we demonstrate that computational complexity is not an obstacle to deployment.
\end{abstract}



\maketitle

\section{INTRODUCTION}
\IEEEPARstart{C}{omplex} Event Processing (CEP) systems are a very prominent big data processing paradigm that helps in detecting patterns, called complex events, in conceptually infinite data streams. The patterns of interest are typically specified in queries with syntax and operators similar to SQL, and specialized operators for complex events, such as the sequence operator. These queries utilize windowing techniques to partition conceptually infinite data streams into finite windows. Incoming data tuples or events are processed immediately, facilitating instant decision-making. CEP systems are applied in various fields, including maritime monitoring, smart building operations~\cite{kumar2024real}, credit card fraud detection, cybersecurity~\cite{alaghbari2022complex}, and the Internet of Things (IoT) \cite{kumar2024rule}. The demand for real-time data stream analysis is increasing due to the growing number of data sources and the valuable and timely insights they can provide \cite{grez2021formal}. Consequently, “virtually all cloud vendors offer first-class support for deploying managed stream processing pipelines” \cite{fragkoulis2024survey}. The data streams processed by CEP often contain sensitive information related to individuals, thereby posing potential privacy risks. Therefore, it is crucial to develop privacy-preserving mechanisms (PPMs) that achieve two seemingly conflicting objectives: ensuring robust privacy protection and maintaining sufficient data utility for applications in use.

Traditional PPMs designed for static datasets cannot be directly used for CEP systems. Multiple studies have proposed PPMs for data streams, regarding all data tuples in a data stream equally important for privacy protection~\cite{fichtenberger2021differentially, dwork2010differential1, chen2017pegasus}. However, not all events for CEP systems are of the same importance for both privacy protection and analysis, as certain events can be part of a complex event that reveals more critical information. We argue that particular privacy protections of these key events can lead to better performance of a privacy-aware CEP system in terms of both privacy protection and data utility compared to state-of-the-art PPMs.

Consider the example of an application that serves taxi drivers and passengers. A continuous data stream of taxi GPS locations is analyzed by a CEP query to find nearby passengers and warn potential traffic jams. The passengers on taxis prefer not to reveal their paths, as they may travel to sensitive locations, e.g., hospitals. In such cases, a common strategy for existing PPMs is to attach sufficient noise to the entire data stream during a trip to conceal the presence at sensitive locations. However, passengers only want to hide their proximity to sensitive locations, i.e., the pattern consisting of a sequence of GPS locations outside and inside sensitive areas within a temporal constrained window. Protecting such patterns instead of all events in a data stream improves data utility by reducing the noise added to the data stream while maintaining equally strong privacy protection. We denote PPMs that protect only particular patterns rather than all events in a data stream as \textbf{pattern-level} PPMs \cite{gu2023differential}. 

In addition to the patterns that contain data producers private information, pattern-level PPMs can also be enhanced by utilizing the patterns of data consumers interests. When these patterns are known, PPMs can be tailored accordingly by adding less noise to the relevant events, thereby improving data utility. We call the analysis conducted by data consumers in such cases \textbf{transparent data stream analysis}, since their analysis based on patterns is transparent to PPMs. In the context of the taxi example, if queries to traffic jam alerts are public, we can identify the events that are more critical for alerts and add less noise to these key events, leading to a higher precision of alerts. However, the patterns of data consumers interests can be confidential. For example, a data consumer that maintains a video-hosting platform might not want to reveal the patterns used by the recommender system of the platform. We call this type of analysis \textbf{opaque data stream analysis}.

Currently, only a few pattern-level PPMs are proposed~\cite{palanisamy2020towards, delouee2023app}. Although they deliver better performance than non-pattern-level PPMs for CEP systems, they are not yet applicable to support most CEP functionalities and limit the performance for privacy-aware CEP systems. To address the shortcomings in existing solutions, we investigate the following research questions in this article:
\begin{itemize}
    \item How to propose a novel privacy guarantee that is valid for multidimensional data and common CEP operators.
    \item How to design PPMs that, for a guaranteed level of privacy protection, minimize redundant noise added to the data and maximize data utility for given private and target patterns.
    \item How to maintain a high-level of data utility when data consumers do not want to reveal their target patterns.
\end{itemize}
Differential privacy (DP) is a state-of-the-art framework for privacy protection with rigorous mathematical proof. Among DP-based models, local differential privacy (LDP) assumes that PPMs are executed before data is released from data producers, leading to a better generalization \cite{ren2022ldp}. Therefore, we propose a novel pattern-level DP that provides pattern-level privacy guarantees and multiple pattern-level PPMs based on LDP for both transparent and opaque data stream analysis. We evaluate the performance of the proposed PPMs and analyze their advantages over the related works. This article makes the following contributions:
\begin{itemize}
    \item We establish a theoretical foundation of pattern-level PPMs and propose a novel pattern-level DP, which enriches both the theories for pattern-level PPMs and DP regarding data streams and supports most CEP operators. It helps to reduce the redundancy for privacy protection in data streams and enables superior performance for PPMs under equally strong privacy guarantee, compared to non-pattern-level PPMs.
    \item We propose four pattern-level PPMs that satisfy pattern-level DP, regarding different context knowledge and trust relationships, i.e., transparent and opaque data stream analysis. 
    \item We evaluate and illustrate the advantages of our proposed PPMs, compared to non-pattern-level state-of-the-art PPMs and analyze their computational complexity.
\end{itemize}

The rest of the article is organized as follows: We introduce related works in \textbf{Section \ref{rw}}. The preliminaries including DP basics and the definition of events and patterns are presented in \textbf{Section \ref{dp}}. We explain the system model, the details of transparent and opaque data stream analysis, and their trust settings in \textbf{Section \ref{system model}}. In \textbf{Section \ref{pldp}}, we propose a novel pattern-level DP guarantee. We present pattern-level PPMs for transparent and opaque data stream analyses in \textbf{Section \ref{plppm}} and \textbf{\ref{ppm for o}}. We evaluate the proposed approaches in \textbf{Section \ref{eva}} and conclude this article in \textbf{Section \ref{con}}.

\section{RELATED WORKS} \label{rw}
It is not surprising that most existing PPMs designed for data stream processing (without the ability to detect complex events) are also applicable to CEP systems, as the input of a CEP system is at least one event stream~\cite{grez2017foundations}. Therefore, we discuss first the most relevant PPMs for data stream processing. Among them, the solutions based on DP are generally considered reliable~\cite{ren2022ldp, kellaris2014differentially}, and three types of DP are most widely discussed, i.e., user-level, event-level, and w-event DP \cite{dwork2010differential1, dwork2010differential2, chen2015private, kellaris2014differentially, chen2017pegasus}. They guarantee privacy by achieving indistinguishability between private and non-private information w.r.t. individual users, individual events, or events in each window for a data stream processing system, respectively \cite{dwork2010differential1, dwork2010differential2, chen2015private, kellaris2014differentially}. They usually transform or snapshot the infinite data stream to a static dataset and utilize DP in that pseudo-dataset~\cite{chen2017pegasus}. In addition to DP, some other studies based on, e.g., k-anonymity, are also proposed~\cite{sopaoglu2021classification, zhou2009continuous}. They can provide satisfactory data utility under guaranteed privacy protection. However, PPMs for data stream processing cannot leverage the potential of transforming the knowledge of complex events into tailored PPMs. The key advantage of PPMs for CEP over those for data stream processing is that the knowledge of private and non-private patterns can be used to avoid redundant privacy protection and unnecessary obfuscation of events which form these patterns. This in turn results in appropriate privacy protection and higher data utility.

Some studies have not fully leveraged the patterns of CEP systems but perform better compared to pure data stream processing PPMs \cite{katsomallos2022landmark, palanisamy2018preserving, wang2020towards}. They emphasize the practical representations of different data tuples and assign different privacy protections to them accordingly. Among them, Landmark privacy \cite{katsomallos2022landmark} assumes that some data tuples may contain significantly more valuable or more private information than others. Therefore, it adjusts the privacy budgets assigned to these tuples to improve its PPMs. However, for a CEP system, there are opportunities for further improvements because the patterns detected by a CEP system are usually much more complex than individual data tuples. Palanisamy et. al. \cite{palanisamy2018preserving} considered privacy protections based on patterns instead of data tuples by reordering detected events. However, it either fully blocks access to a pattern or publishes the pattern without any privacy protection. Therefore, it provides much weaker granularity than Landmark and other DP-based studies, and thus much weaker support for CEP operators. 

Unlike pattern-level PPMs for data stream processing systems, there are also pattern-aware approaches that focus on statistically significant sequences of data tuples \cite{wang2020towards,mao2024privshape,hu2024real}. These sequences are referred to as patterns in the related works, which are, however, distinct from the patterns defined for CEP systems and in this article. Two pattern-aware approaches based on DP are presented for data collection (PatternLDP)~\cite{wang2020towards} and pattern extraction (PrivShape)~\cite{mao2024privshape}. They usually process one-dimensional time series that are of the same data type, which is untypical for CEP systems that usually take multidimensional data streams of different attributes as input. Furthermore, patterns for CEP systems are constructed based on concrete applications and practical representations, which are not necessarily statistically significant. In conclusion, patterns as defined in the context of CEP systems do not correspond to the notion of patterns presented in these pattern-aware approaches. Another pattern-aware approach (RetraSyn) presents real-Time trajectory synthesis which focuses on preserving spatial-temporal patterns under LDP \cite{hu2024real}. Although it outperforms other state-of-the-art pattern-aware solutions on trajectory data, it is affected by the same disadvantages as focusing on statistically significant patterns and cannot be adapted to CEP systems.

Our previous work presented a pattern-based solution with DP guarantees for CEP systems~\cite{gu2023differential}. It increases data utility by adaptively assigning privacy budgets over events related to private patterns. However, this solution is limited because (1) it can only be applied to patterns formed by the sequence operator and (2) it assumes that all queries from data consumers are public, which may not be practical, since the queries can be confidential considering the business model of consumers. To the best of our knowledge, there is no solution that optimizes data utility and (1) provides DP-based privacy protection for private patterns, (2) supports all common operators of CEP systems, and (3) allows data consumers to keep their patterns of interest confidential.

\section{PRELIMINARIES} \label{dp}
In this section, we introduce the basics of DP for databases and data streams as the fundamentals of developing pattern-level DP. We also clarify the definitions of events and patterns needed to present the pattern-level DP and its PPMs.
\subsection{DIFFERENTIAL PRIVACY}
A mechanism that complies with DP ensures that an adversary cannot distinguish two similar groups of information based on the responses of queries to them. DP formulates a bound on the information that an adversary with arbitrary side-information and computational power can reveal.
\begin{definition} [$\epsilon$-differential privacy]
    A privacy mechanism $\mathcal{M}$: $\mathcal{I}$ $\rightarrow$ $\mathcal{O}$, where $\mathcal{I}$ is an arbitrary input dataset endowed with a neighboring relation N, is said to be $\epsilon$-differentially private if for any $\mathcal{O}_i \subseteq \mathcal{O}$,
    $$\Pr[\mathcal{M}(I_i) \in \mathcal{O}_i] \leq 
    e^{\epsilon} \cdot \Pr[\mathcal{M}(I'_i) \in \mathcal{O}_i], \hspace{0.3cm} 
    \forall I_i N I'_i,$$
    where $\epsilon$ is called \textbf{privacy budget}.
\end{definition}
The neighboring relation between $I_i$ and $I'_i$ indicates that $I'_i$ can be generated from $I_i$ by replacing, adding, or removing only one record from $I_i$. This ordinary DP based on static datasets can be modified to handle dynamically updated datasets, i.e., data streams \cite{kellaris2014differentially}. The neighboring relation is then modified to be the adjacency of data streams. It means that for two infinite data streams $S^D = (d_1, d_2, ...)$ and $S^{D'} = (d'_1, d'_2, ...)$  of data tuples $d$, there exists one and only one unique $i$ such that $d_i \neq d'_i$, leading to the following definition of the $\epsilon$-differential privacy for data streams \cite{dwork2010differential1}.
\begin{definition} [$\epsilon$-differential privacy for data streams]
    For an arbitrary infinite data stream $S^D$ and its neighboring stream $S^{D'}$, a privacy mechanism $\mathcal{M}$ that takes a data stream as the input and outputs $\mathcal{O}_i \subseteq \mathcal{O}$ is $\epsilon$-differentially private if for any $\mathcal{O}_i \subseteq \mathcal{O}$, it holds that $$\Pr[\mathcal{M}(S^D) \in \mathcal{O}_i] \leq 
    e^{\epsilon} \cdot \Pr[\mathcal{M}(S^{D'}) \in \mathcal{O}_i].$$
\end{definition}
DP for data streams guarantees indistinguishability between outputs when taking neighboring data streams as inputs. By modifying the definition of the neighboring relation, one may propose diverse DP guarantees with different granularities of privacy protections, different practical representations, and different application scenarios.
\begin{figure*}
\centering
    \subfloat[Privacy model for transparent data stream analysis\label{modela}]{\includegraphics[width=0.47\textwidth]{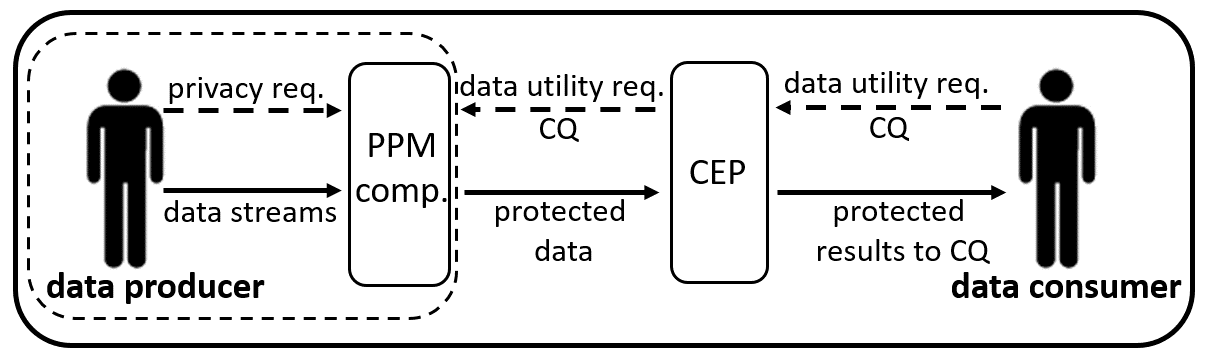}}%
    \subfloat[Privacy model for opaque data stream analysis\label{modelb}]{\includegraphics[width=0.47\textwidth]{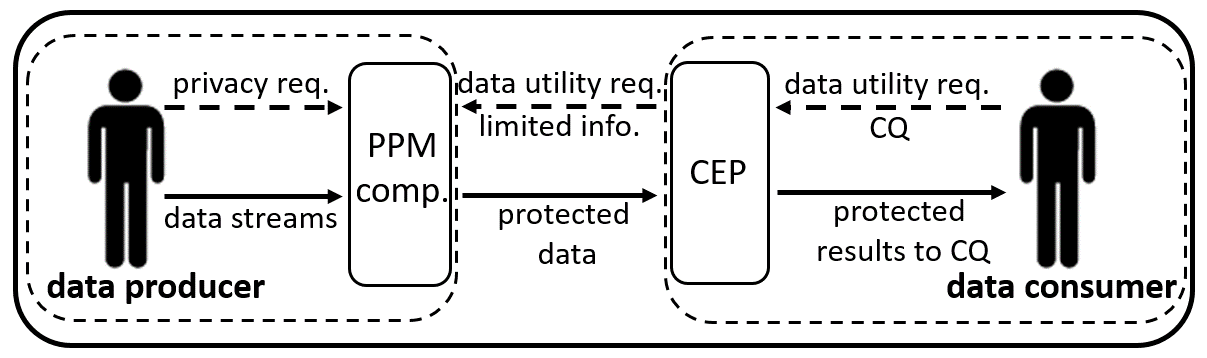}}%
\caption{Workflow of transparent and opaque data stream analysis. -{ }-{ }-$>$: the setup phase; \textemdash$>$: the service provision phase; Rectangles in dashed lines: trust domains; \textit{CQ}: continuous queries; \textit{req.}: requirements; \textit{info.}: information.}
\vspace{-7pt}
\label{all_privacy_models}
\end{figure*}

\subsection{EVENTS AND PATTERNS} \label{e and p}
It is a common understanding that events comprise timestamps and a set of attribute-value pairs. For better comprehensibility, this article uses a rudimentary case in which each event contains one attribute-value pair with one timestamp. It does not restrict the generalization of our approach, because any arbitrary event with $n$ attribute-value pairs can be regarded as $n$ events with one attribute-value pair. The values of events can be numerical or categorical. Extracting all events from a data stream forms an event stream $S^E = (e_1, e_2, ...)$, where $e_i$ denotes the $i^{th}$ produced event. Multiple data streams from different sources can be merged into one event stream, where $e_i$ is the $i^{th}$ extracted event. Within an event stream, an ordered sequence of events may contain valuable information that we denote as a pattern $P = Q(e^Q_1, e^Q_2, ..., e^Q_m)$, where $Q$ is a query formed by data stream operators that can reveal the information contained in these events. A pattern stream $S^P = (P_1, P_2, ...)$ can then be formed by extracting all patterns from an event stream, where $P_i$ is the $i^{th}$ detected pattern. Some detected patterns may contain private information of data producers and are classified as \textbf{private patterns}, while all others are \textbf{public patterns}. Furthermore, for a data consumer, the patterns that are of their interest are their \textbf{target patterns}. 

In the context of the Taxi example, data streams contain GPS records of all taxis with passengers. We collect the records of interest and merge them into an event stream. The sequence of several events forms a pattern, e.g., a taxi drives to a hospital. Different combinations of events may form the same type of patterns. For example, both when a taxi approaches a hospital from north and south are the same type of patterns that the taxi drives to the hospital. Therefore, we see a demand to distinguish specific patterns from the type of patterns. We formally clarify that a \textbf{pattern} $P = Q(e^Q_1, e^Q_2, ..., e^Q_m)$ referred to in this article is a specific combination of events $e^Q_i$ that can be retrieved by query $Q$ in a given event stream. Meanwhile, a \textbf{pattern type} $\mathcal{P}$ is a group of patterns specified by a query $Q_\mathcal{P}$. A pattern $P_i$ is an element of $\mathcal{P}$, i.e., $P_i \in \mathcal{P}$, if and only if it can be retrieved by $Q_\mathcal{P}$. Since pattern-level PPMs focus on patterns instead of basic events, their performance metric w.r.t. data utility can be a combination of both recall and precision of detecting target patterns in data streams. Inspired by statistical classification methods, we introduce $F_1$-score as the performance metric of PPMs w.r.t. data utility~\cite{sasaki2007truth}:
\begin{equation}
    U = F_1 = 2\frac{\mathit{precision} \cdot \mathit{recall}}{\mathit{precision} + \mathit{recall}}.
\end{equation}

\section{SYSTEM MODEL AND TRUST SETTINGS}\label{system model}
This work is based on the concept of LDP and a system model consisting of four basic components (see Figure~\ref{all_privacy_models}): 
\begin{itemize}
    \item \textbf{Data producers} share their data, specify their private patterns, and demand that they be protected. 
    \item The \textbf{PPM component} deploys PPMs to protect private patterns and maintain sufficient data utility.
    \item The \textbf{CEP engine} continuously analyzes privacy-protected data streams received from the PPM component and forwards the detected target patterns to data consumers. 
    \item \textbf{Data consumers} submit queries that describe the target patterns to the CEP engine and declare their requirements for data utility needs. Data consumers are potential privacy adversaries, following the \textit{honest-but-curious} threat model. 
\end{itemize}
In Figure~\ref{all_privacy_models}, the rounded rectangles in dashed lines denote the trust domains where the original information within is by default not accessible by components outside the domain.

The system operates in two distinct phases: the set-up phase and the service provision phase. During the set-up phase (illustrated by dashed arrows in Figure~\ref{all_privacy_models}) data producers specify their privacy requirements, i.e., private patterns while data consumers specify queries for their target patterns and their data utility requirements. This information enables the PPM component to configure a tailored PPM that meets the specified requirements. Additionally, the data consumer registers the queries for the target patterns in the CEP component. In the service provision phase (illustrated by solid arrows in Figure~\ref{all_privacy_models}), original data from the data producer is streamed to the PPM, which protects the private patterns while preserving the required utility of data elements forming the target patterns. Subsequently, the output of the PPM is streamed to the CEP component, which continuously processes the consumers’ queries against the incoming data stream to detect the specified target patterns. All detected target patterns are then forwarded to the data consumer.

Given that this work builds on the principals of LDP, the PPM component is consistently deployed in the trust domain of the data producer. Since the output of the PPM is privacy protected, the data producer does not need to establish any trust relation with the CEP component. However, the relationship with the data consumer can vary, leading to two distinct scenarios: (1) the data consumer is unconcerned sharing the target pattern queries with the PPM component and an untrusted CEP engine or, (2) the data consumer prefers to conceal target pattern queries to untrusted entities. For instance, in a video hosting platform utilizing a machine learning based recommendation system, confidentiality is vital for its business model. In such cases, the data consumer does not provide the target patterns to the PPM components and requires a trusted CEP component to register and process the target pattern queries. When target pattern queries are visible to the PPM component, this is called \textbf{transparent data stream analysis} (Figure~\ref{modela}); conversely, if target pattern queries are hidden, it is referred to as \textbf{opaque data stream analysis} (Figure~\ref{modelb}). In case of transparent data stream analysis the PPM component leverages the knowledge of private and target patterns to optimally tailor the PPM. In opaque data stream analysis, however, the PPM lacks access to target patterns, limiting its capacity for optimization. Nonetheless, rather than entirely forgoing the insights obtainable from data consumers, we propose a federated approach between the PPM component and data consumers that does not require disclosing target patterns. In the simplest case, the data consumers provide the number of events that have to be analyzed to detect a target pattern, called the \textbf{observation span} of a target pattern. With this limited information, our pattern-level PPMs can still significantly enhance data utility performance compared to other PPMs, providing a compelling incentive for data consumers to collaborate with the PPM component.

\section{PATTERN-LEVEL DIFFERENTIAL PRIVACY}\label{pldp}
In this section, we propose a novel DP guarantee, which achieves $\epsilon$-DP w.r.t. patterns and is hence named pattern-level DP. Its design consists of three steps: 
\begin{itemize}
    \item We define the neighboring relation between individual patterns, i.e., in-pattern neighbors (Definition 3), because it forms the fundamental of the neighboring relation between pattern streams.
    \item We define the neighboring relation between pattern streams, i.e., pattern-level neighbors (Definition 4), because it affects the characteristics and granularity of the privacy protection guaranteed by pattern-level DP.
    \item We propose the pattern-level DP (Definition 5) in pattern streams for a given pattern type $\mathcal{P}$, based on the defined neighboring relations.
\end{itemize}

\begin{definition}[in-pattern neighbors]
    Two patterns $P = Q(e^Q_1, e^Q_2,$ $ ..., e^Q_m),$ and $P' = Q'(e^{Q'}_1, e^{Q'}_2, ..., e^{Q'}_m),$ of the same pattern type $\mathcal{P}$ and the same length are \textbf{in-pattern neighbors of $\mathcal{P}$} if and only if there exists a unique $i$ such that (1)~$e^Q_i$ and $e^{Q'}_i$ only differ by their attribute-value pairs and (2) for all $j \neq i$, $e^Q_j = e^{Q'}_j$.
\end{definition}
In-pattern neighboring relation clarifies the least possible difference between two most similar patterns. This means that two neighboring patterns of the same pattern type $\mathcal{P}$ can only differ by one single event, which makes them one of the most similar patterns to each other in $\mathcal{P}$. Given Definition 3, we can define the neighboring relation for pattern streams.
\begin{definition}[pattern-level neighbors]
    Given a pattern type $\mathcal{P}$, and two infinite pattern streams $S^P = (P_1, P_2, ...)$ and $S^{P'} = (P'_1, P'_2, ...)$, $S^P$ and $S^{P'}$ are \textbf{pattern-level neighbors} w.r.t. $\mathcal{P}$ if and only if there exists a unique integer $i$ such that (1) $P_i$ and $P'_i$ are in-pattern neighbors of $\mathcal{P}$, and (2) for all $j \neq i$, $P_j = P'_j$ holds.
\end{definition}
The pattern-level neighboring relation clarifies the granularity of the indistinguishability guaranteed by the pattern-level DP. In practice, this neighboring relation describes two distinguishable yet least different pattern streams for a typical CEP system because two pattern-level neighboring streams only differ by the neighboring patterns, which differ by one basic event. The pattern-level DP is then proposed as follows.
\begin{definition}[pattern-level DP]
    Assume that $\mathcal{M}$ is a privacy mechanism that takes a pattern stream as input and outputs a response $R$ that belongs to the group of all possible responses $\mathcal{R}$. Then $\mathcal{M}$ satisfies \textbf{pattern-level} $\boldsymbol{\epsilon}$\textbf{-DP} of a given pattern type $\mathcal{P}$ (pattern-level DP of $\mathcal{P}$) if and only if for any pattern-level neighbors $S^P$ and $S^{P'}$ of $\mathcal{P}$ and any sets of response $\mathcal{R}_i \subseteq \mathcal{R}$, 
    $$\Pr[\mathcal{M}(S^P) \in \mathcal{R}_i] \leq 
    e^{\epsilon} \cdot \Pr[\mathcal{M}(S^{P'}) \in \mathcal{R}_i].$$
\end{definition}
Pattern-level DP of $\mathcal{P}$ guarantees privacy for $P_i \in \mathcal{P}$. In practice, $\mathcal{P}$ are usually private pattern types of data producers. Mechanisms that fulfill pattern-level DP output similar query results regardless of the existence of private patterns. Pattern-level DP provides a more customized privacy protection compared to non-pattern-level DPs, which allows PPMs to reduce the privacy budgets $\epsilon$ assigned to less important events and utilize the budgets more efficiently, leading to better privacy-utility trade-off and overall superior performance. Although a precise definition of target patterns is helpful in improving data utility, the overall privacy budgets assigned to private patterns are fixed. This indicates that even if the data utility is impaired by imprecise target patterns, the privacy protection is not weakened. For target patterns, there is no limitation for the CEP operators used to form these patterns. However, for private patterns, the used operators must be able to catch the difference between patterns caused by their different belonging events. In other words, private patterns must have in-pattern neighbors and must be able to form pattern-level neighbors. Most CEP operators, e.g., Max/Min operators, average operators, and followed by operators, satisfy this requirement.

\section{PATTERN-LEVEL PPMS FOR TRANSPARENT DATA STREAM ANALYSIS}\label{plppm}
In Section \ref{dp}.\ref{e and p}, we classify all events into two categories, i.e., categorical events and numerical events. To achieve DP, we usually attach noise to individual events. However, the noise that is appropriate for numerical events, e.g., age, is usually improper for categorical events, e.g., nationality. Therefore, for different types of events, we attach different types of noise to protect their contained private information. In detail, we present the \textbf{randomized response} for categorical events, which provides a random selection among all possible categories with a set of assigned probabilities. For numerical events, we apply the \textbf{Laplace mechanism}, which adds Laplace noise to events based on a Laplace distribution. The randomized response and Laplace mechanism are proven to be differentially private \cite{dwork2014algorithmic}. However, this article presents a novel DP guarantee, i.e., pattern-level DP, which is distinct from the ordinary DP. Therefore, in order to utilize the randomized response and Laplace mechanism to achieve pattern-level DP, their pattern-level DP characteristics must also be verified, apart from their ordinary DP characteristics. In such cases, they are also required to be modified to match pattern-level privacy protection measures. We first propose the randomized response for pattern streams and verify its pattern-level DP characteristics, followed by the definition, verification, and proofs for the Laplace mechanism.

\begin{definition}[Randomized response (pattern streams)]
Given a pattern stream $S^P = (P_1, P_2, ...)$ that consists of patterns $P = Q(e^{Q}_1, e^{Q}_2, ..., e^{Q}_m)$, then a privacy mechanism $\mathcal{M}$ offers a randomized response if it takes the existence of events $I(e^{Q_i}_i) \in \{0, 1\}$ as input and outputs responses $R_i \in \{0, 1\}$ for each event with probability where $j \neq k$:
\begin{equation*}
    \begin{cases}
        \Pr(R_i = j | I(e_i) = j) = 1- p_i & \\
        \Pr(R_i = j | I(e_i) = k) = p_i, & \\
    \end{cases}
\end{equation*}
\end{definition}
To verify the pattern-level DP characteristics of the randomized response, we start with the DP characteristics w.r.t. the events of a given pattern in a pattern stream. Subsequently, we derive the result to an entire pattern stream.
\begin{lemma}
Given a pattern stream $S^P = (P_1, P_2, ...)$ consisting of patterns $P = Q(e^{Q}_1, e^{Q}_2, ..., e^{Q}_m)$, for a privacy mechanism $\mathcal{M}_1$ that (1) offers a randomized response for $e_i$ with $p_i \leq \frac{1}{2}$, (2) takes $\boldsymbol{I}(\boldsymbol{e}) = (I(e_1), I(e_2), ...,I(e_n))$ as inputs, and (3) outputs responses $\boldsymbol{R} = (R_1, R_2, ..., R_n)$, it guarantees $\ln{\frac{1-p_i}{p_i}}$-pattern-level DP w.r.t. a given type of pattern $\mathcal{P}$.
\end{lemma}
\begin{proof}
For two pattern-level neighbors $S^P = (P_1, P_2, ...)$ and $S^{P'} = (P'_1, P'_2, ...)$ of pattern type $\mathcal{P}$, if given any randomized response mechanism $\mathcal{M}_1$ as described above, then
\begin{equation*}
\frac{\Pr[\mathcal{M}_1(S^P) \in \mathcal{R}]}{\Pr[\mathcal{M}_1(S^{P'}) \in \mathcal{R}]} = {\frac{\Pr[\mathcal{M}_1(I(e_i)) = R_i]}{\Pr[\mathcal{M}_1(I(e_i')) = R_i]}} \leq {\frac{1-p_i}{p_i}}. 
\end{equation*}
\end{proof}
It proves that $\mathcal{M}_1$ guarantees $\ln{\frac{1-p_i}{p_i}}$-pattern-level DP for the above settings. We subsequently study the pattern-level DP for any pattern in a pattern stream.
\begin{theorem}
Given a pattern stream $S^P = (P_1, P_2, ...)$ consisting of patterns, a privacy mechanism $\mathcal{M}$ that offers a randomized response for $P = Q(e^{Q}_1, e^{Q}_2, ..., e^{Q}_n)$ that belongs to a given pattern type $\mathcal{P}$ with $p_1, p_2, ..., p_n \leq \frac{1}{2}$. If it takes $\boldsymbol{I}(\boldsymbol{e}) = (I(e_1), I(e_2), ...,I(e_n))$ as inputs and outputs responses $\boldsymbol{R} = (R_1, R_2, ..., R_n)$, $\mathcal{M}$ guarantees $\sum_{i:e_i \in P}{\ln{\frac{1-p_j}{p_j}}}$-pattern-level DP with respect to a given type of pattern $\mathcal{P}$.
\end{theorem}
\begin{proof}
For two pattern-level neighbors $S^P = (P_1, P_2, ...)$ and $S^{P'} = (P'_1, P'_2, ...)$ of pattern type $\mathcal{P}$, if given any randomized response mechanism $\mathcal{M}$ as described above, then
\begin{equation*}
\begin{split}   
\frac{\Pr[\mathcal{M}(S^P) \in \mathcal{R}]}{\Pr[\mathcal{M}(S^{P'}) \in \mathcal{R}]} 
& = \prod_{i:e_i \not\in P}{\frac{\Pr[\mathcal{M}(I(e_i)) = R_i]}{\Pr[\mathcal{M}(I(e_i')) = R_i]}} 
\\& \quad \cdot \prod_{i:e_i \in P}{\frac{\Pr[\mathcal{M}(I(e_i)) = R_i]}{\Pr[\mathcal{M}(I(e_i')) = R_i]}} 
\\& \leq \prod_{j:e_i \in P}{\frac{1-p_i}{p_i}}. 
\end{split}
\end{equation*}
\end{proof}
It proves that $\mathcal{M}$ guarantees $\ln{\prod_{j:e_j \in P}{\frac{1-p_j}{p_j}}}$-pattern-level DP, i.e., $\sum_{j:e_j \in P}{\ln{\frac{1-p_j}{p_j}}}$-pattern-level DP. Therefore, applying randomized response to related events leads to pattern-level differentially private outputs. We then present the Laplace mechanism and verify that it satisfies pattern-level DP. We first introduce an important concept of Laplace mechanisms, i.e., the sensitivity of a privacy mechanism \cite{dwork2014algorithmic} and define Laplace mechanisms.
\begin{definition} [Sensitivity of a privacy mechanism]
    Given a privacy mechanism $\mathcal{M}$ that takes $e_i$ as inputs and outputs responses $R_i$. The sensitivity of $\mathcal{M}$ is then defined as $\mathbb{S(\mathcal{M})} \triangleq \sup|R_i-R_j|$, where $e_i$ and $e_j$ are neighboring.
\end{definition}
\begin{definition}[Laplace mechanisms (pattern streams)]
    Given a pattern stream $S^P = (P_1, P_2, ...)$ that consists of patterns $P = Q(e^{Q}_1, e^{Q}_2, ..., e^{Q}_n)$, a privacy mechanism $\mathcal{M_L}$ is a Laplace mechanism, if it takes $\boldsymbol{I}(\boldsymbol{e}) = (I(e_1), I(e_2), ...,I(e_n))$ as inputs and outputs responses $\boldsymbol{R} = (R_1, R_2, ..., R_n)$, where $R_i = I(e_i) + \omega_i$, and $\omega_i \sim \mathit{Laplace}_i$.
\end{definition}
We study the pattern-level DP characteristics of the Laplace mechanism following the same procedure of the randomized response mechanism.
\begin{theorem}
    For two pattern-level neighbors $S^P = (P_1, P_2, ...)$ and $S^{P'} = (P'_1, P'_2, ...)$ of a pattern type $\mathcal{P}$, a Laplace mechanism $\mathcal{M_L}$ applied to any pattern $P \in \mathcal{P}$ with sensitivity $\mathbb{S(\mathcal{M})}$ guarantees ${\sum_{i:e_i \in P}{{\mathbb{S(\mathcal{M})}}/\lambda_j}}$ pattern-level DP.
\end{theorem}
\begin{proof}
For two pattern-level neighbors $S^P = (P_1, P_2, ...)$ and $S^{P'} = (P'_1, P'_2, ...)$ of a pattern type $\mathcal{P}$, if given any randomized response mechanism $\mathcal{M}$ as described above, then we see that  
\begin{equation*}
\centering
\begin{split}
\frac{\Pr[\mathcal{M_L}(S^P) \in \mathcal{R}]}{\Pr[\mathcal{M_L}(S^{P'}) \in \mathcal{R}]} 
& = \prod_{j:e_j \not\in {P}}{\frac{\Pr[\mathcal{M_L}(I(e_j)) = R_j]}{\Pr[\mathcal{M_L}(I(e_j')) = R_j]}} 
\\ & \quad \cdot 
\prod_{i:e_i \in {P}}{\frac{\Pr[\mathcal{M_L}(I(e_i)) = R_i]}{\Pr[\mathcal{M_L}(I(e_i')) = R_i]}} 
\\ &  \leq 
\prod_{i:e_i \in {P}}{e^{{\mathbb{S(\mathcal{M})}}/\lambda_i}} = e^{\sum_{i:e_i \in {P}}{{\mathbb{S(\mathcal{M})}}/\lambda_i}}.
\end{split}
\end{equation*} 
\end{proof}
It is proven that $\mathcal{M}$ guarantees ${\sum_{i:e_i \in {P}}{{\mathbb{S(\mathcal{M})}}/\lambda_i}}$-pattern-level DP, and hence, both the randomized response and Laplace mechanism satisfy pattern-level DP. Based on them, we subsequently present two pattern-level PPMs for transparent and opaque data stream analysis, i.e., a uniform approach and a bidirectional step-wise approach.

\subsection{THE UNIFORM APPROACH}
The previous sections show that the privacy budgets assigned to each event, i.e., $\epsilon$ of pattern-level DP, eventually accumulate in the pattern stream. Therefore, with a given total privacy budget $\epsilon$, we can distribute it to all \textbf{relevant events} which belong to private patterns and can be revealed by data consumers when querying target patterns. This leads to the most critical difference in the budget distribution between pattern-level DP and the DPs proposed by the related works. In the related works, the privacy budgets are usually assigned among a sequence of events listed with temporal orders. For example, regarding the w-event DP with a window of size $w$ \cite{ren2022ldp}, the privacy budgets are distributed among the events from timestamp $t$ to timestamp $t + w - 1$. The specific distribution can be improved to reach a better data utility, e.g., by dividing all data producers into groups and arranging the timestamp of releasing the data of each data producer to enhance the privacy-utility trade-off \cite{ren2022ldp}. Unlike the related works, the pattern-level DP assigns privacy budgets only to those events that are relevant to private and target patterns, which enables a more flexible and adaptive privacy budget distribution than, e.g., w-event privacy.

An intuitive approach is to uniformly distribute $\epsilon$, which is named the \textit{uniform} approach. Given a private pattern $P = Q(e^{Q}_1, e^{Q}_2, ..., e^{Q}_n)$, we denote the privacy budget assigned to $e^{Q}_i$ as $\epsilon_i$. For a randomized response, we have proven that $\epsilon_i = \ln{\frac{1-p_i}{p_i}}$, while for a Laplace mechanism, $\epsilon_i = \mathbb{S(\mathcal{M})}/\lambda_i$ holds. Therefore, the probabilities $p_i$ for the randomized response and $\lambda_i$ for the Laplace mechanism are the only parameters relevant to the privacy budget distribution, when the sensitivity $\mathbb{S(\mathcal{M})}$ is constant for each event. In practice, the sensitivity is determined by the boundary values of the attributes in the given data stream, which are usually fixed. We can therefore regard the sensitivity as a constant, and the problem of assigning privacy budgets is then equivalent to determining $p_i$ and $\lambda_i$. Here, $p_i$ controls the probability of obtaining correct responses when querying target patterns, while $\lambda_i$ controls the amount of noise attached to an event. Therefore, they are also the only parameters that affect the data utility, in terms of $\mathit{MRE}_U$. Optimizing $p_i$ and $\lambda_i$ is then the only problem in optimizing the performance of PPMs. 

Furthermore, since both $\lambda_i$ and $p_i$ are independent of $\epsilon_j$ if $i \neq j$, the privacy budget distributions for numerical events and categorical events are also independent from each other, if we do not employ both types of mechanism on the same event. It indicates that one combination of privacy budgets corresponds to one combination of $\lambda_i$ and $p_i$. We can therefore employ the Laplace mechanisms for numerical events and meanwhile the randomized response for categorical events, without considering the overlapping effects between them. In conclusion, for a given trust relationship, the optimal $\epsilon_i$ will lead to optimal sets of $\lambda_i$ and $p_i$ and hence to the optimal performance of PPMs. In such cases, although multiple independent private patterns overlap each other, i.e., share certain events, their privacy budget distributions are still independent. In other words, their shared events are assigned different and independent privacy budgets multiple times for each private pattern to which the events belong. 

\subsection{THE BIDIRECTIONAL STEP-WISE APPROACH}
The \textit{uniform} approach considers all relevant events equally important. However, some private events may be critical for detecting target patterns, while containing little private information. It is more beneficial to assign more privacy budgets to these events, leading to weaker privacy protection but higher data utility. Inspired by statistical learning, we propose a bidirectional step-wise algorithm, i.e., the \textit{step-wise} approach, that utilizes historical data to optimize $\epsilon_i$, as demonstrated in Algorithm 1. Algorithm 1 is executed only before starting a data stream. This indicates that the privacy budgets assigned to the events of the same type remain unchanged for a run-time data stream. The data utility can be improved by executing Algorithm 1 before streaming the data, because (1) the schema, (2) the compositions of privacy patterns, and (3) the compositions of target patterns in a flowing data stream are always fixed.

We first evenly distribute the privacy budget to all relevant events that belong to private patterns and can be revealed by data consumers while querying target patterns. After creating an ordered list for these events, we increase the privacy budget assigned to the first event while equally reducing the budgets for other events. If the data utility increases, we maintain this change and otherwise abandon it. We iterate this procedure for each event in the list for multiple times, until the available computation time runs out or the data utility reaches a predefined expectation. In ideal cases, Algorithm 1 is capable of finding the optimal $\epsilon_i$, as well as the optimal $p_i$ and $\lambda_i$, and provides significantly superior performance than the \textit{uniform} approach. In practice, Algorithm 1 may not reach the theoretical optimum, but its performance can converge to the optimum by consuming more up-to-date field data. There are two critical assumptions for Algorithm 1: (1) The expected statistical distribution of the collected historical data also applies to the coming data in the pattern stream. (2) The schema of historical data and the future pattern streams strictly match each other. Infractions to them can degrade performance.
\begin{algorithm}
\begin{algorithmic}
 \STATE 1. Distribute the total privacy budget $\epsilon$ evenly to all the $m$ events that are relevant to private patterns $\epsilon_i = \frac{\epsilon}{m}$.\\
 \STATE 2. Select the size of each step based on field experience. A suggestion is $\delta_\epsilon = \frac{\epsilon}{100m}$.\\
 \STATE 3. Calculate data utility $U$; Set $U_1 = U_2 = ... = U_m = U$.\\
 \WHILE{$\epsilon_1, \epsilon_2,...,\epsilon_m  \in [0, \epsilon]$ and $U_{i:U_i = \max{U_i}} \geq U$}{
    \STATE 4.1. Set $i = 1$. Set $U = U_{i:U_i = \max{U_i}}$.\\
    \WHILE{$i \leq m$}{
        \STATE 4.2.1. Set $\epsilon_i = \epsilon_i + \delta_\epsilon$, and set all other privacy budgets $\epsilon_{j:j \neq i} = \epsilon_{j} - \frac{\delta_\epsilon}{m-1}$.\\ 
        \STATE 4.2.2. Calculate data utility metric $U_i$. \\
        \STATE 4.2.3. Set $\epsilon_i = \epsilon_i - \delta_\epsilon$, and set all other privacy budgets $\epsilon_{j:j \neq i} = \epsilon_{j} + \frac{\delta_\epsilon}{m-1}$. Set i = i + 1.\\
    }
    \ENDWHILE
    \IF{$U_{i:U_i = \max{U_i}} \geq U$}{
        \STATE 4.3. Set $\epsilon_{i:U_i= \max{U_i}} = \epsilon_i + \delta_\epsilon$. Set all other privacy budgets $\epsilon_{j:j \neq i} = \epsilon_{j} - \frac{\delta_\epsilon}{m-1}$.\\
    }
    \ENDIF
    \IF{computation time runs out}{
        \STATE 4.4. \textbf{break}.
    }
    \ENDIF
 }
\ENDWHILE
\caption{Bidirectional step-wise privacy budget distribution for transparent data stream analysis}
\end{algorithmic}
\end{algorithm}

\section{PATTERN-LEVEL PPMS FOR OPAQUE DATA STREAM ANALYSIS}\label{ppm for o}
For opaque data stream analysis, the PPM component is not informed about queries to target patterns. Therefore, additional measures are required to compensate for the missing information. In this section, we adapt the proposed PPMs for transparent data stream analysis so that they can be constructed with limited information about target patterns. The adaptation fulfills the requirements presented in Section \ref{system model}, i.e., the adapted PPMs do not reveal the queries that retrieve target patterns.

\subsection{THE UNIFORM APPROACH}
Although the structures of target patterns are unknown to the PPM component, the \textit{uniform} approach is not severely affected, since it requires only the observation span of target patterns, as defined in Section \ref{dp}.\ref{e and p}. Once the actual observation span is longer than the estimated value, redundant privacy protections are applied, leading to unnecessarily lower data utility. In practice, we may assume that the actual observation span is fixed as $l$, while the observation span provided by consumers is $\hat{l}$. When $\hat{l} < l$, we would apply privacy protection more frequently than we should, leading to redundant privacy protection and overall fewer privacy budgets. For a given theoretical privacy budget $\hat{\epsilon}$, the actual privacy budget applied for privacy protection is $\epsilon = \frac{\hat{\epsilon}\hat{l}}{l}.$ Assuming that data consumers follow the \textit{honest-but-curious} model and provide an accurate $\hat{l}$, we can then assign the privacy budgets uniformly. Since only an approximate observation span of a target pattern is provided to the PPM component, it is impossible to infer the concrete structure of the target pattern. Therefore, the target pattern and its queries are protected for the adapted \textit{uniform} approach.

\subsection{THE BIDIRECTIONAL STEP-WISE APPROACH} \label{opaque_bidirectional}
In addition to the observation spans, the \textit{step-wise} approach only requires data utilities from data consumers, since the $F_1$-scores can only be provided by consumers once the target patterns are unknown. An intuitive approach is to provide both the original data and protected data to data consumers to calculate the decrease in data utility. If the collected historical data are provided by data consumers, there are no privacy question marks. However, if the historical data are not owned by data consumers, they can still be private against the consumers. An obfuscation procedure is then employed for privacy protection. When processing historical data, real-time detection of target patterns is not required. Therefore, for each historical dataset, we generate $n$ obfuscation datasets with distinct yet sufficient noise attached. All obfuscation datasets are delivered to data consumers with the original dataset, expecting data utilities in return. In such cases, the privacy protection of the historical data can reach any predefined intensity, since as $n$ tends to infinity, the privacy protection tends to be infinity strong. We conclude the \textit{step-wise} approach for opaque data stream analysis in Algorithm 2. Note that the approach will not reveal target patterns, as neither of its required information can be a threat.

\begin{algorithm}
\begin{algorithmic}
 \STATE 1. For $m$ events that are relevant to private patterns, and a given total privacy budget $\epsilon = \frac{\hat{\epsilon}\hat{l}}{l}$, distribute evenly the budget $\epsilon_i = \frac{\epsilon}{m}$.\\
 \STATE 2. Select the size of each step based on field experience. A suggestion is $\delta_\epsilon = \frac{\epsilon}{100m}$.\\
 \STATE 3. Calculate data utility $U$; Set $U_1 = U_2 = ... = U_m = U$.\\
 \WHILE{$\epsilon_1, \epsilon_2,...,\epsilon_m  \in [0, \epsilon]$ and $U_{i:U_i = \max{U_i}} \geq U$}{
    \STATE 4.1. Set $i = 1$. Set $U = U_{i:U_i = \max{U_i}}$.\\
    \WHILE{$i \leq m$}{
        \STATE 4.2.1. Set $\epsilon_i = \epsilon_i + \delta_\epsilon$, and set all other privacy budgets $\epsilon_{j:j \neq i} = \epsilon_{j} - \frac{\delta_\epsilon}{m-1}$.\\ 
        \STATE 4.2.2. Deliver datasets with $n$ obfuscation datasets to data consumers to calculate data utility $U_i$.\\
        \STATE 4.2.3. Set $\epsilon_i = \epsilon_i - \delta_\epsilon$, and set all other privacy budgets $\epsilon_{j:j \neq i} = \epsilon_{j} + \frac{\delta_\epsilon}{m-1}$. Set i = i + 1.\\
    }
    \ENDWHILE
    \IF{$U_{i:U_i = \max{U_i}} \geq U$}{
        \STATE 4.3. Set $\epsilon_{i:U_i= \max{U_i}} = \epsilon_i + \delta_\epsilon$, and set all other privacy budgets $\epsilon_{j:j \neq i} = \epsilon_{j} - \frac{\delta_\epsilon}{m-1}$.\\
    }
    \ENDIF
    \IF{computation time runs out}{
        \STATE 4.4. \textbf{break}.
    }
    \ENDIF
 }
 \ENDWHILE
 \caption{Bidirectional step-wise privacy budget distribution for opaque data stream analysis}
\end{algorithmic}
\end{algorithm}

\section{EVALUATION}\label{eva}
To evaluate the proposed PPMs, we perform three studies. Study (1): A comparative study to compare our PPMs with the most relevant related works, i.e., budget distribution (BD) \cite{kellaris2014differentially} and Landmark privacy \cite{katsomallos2022landmark}, as well as two recent proposals based on statistical definitions of patterns, i.e., PrivShape in classification task mode\footnote{Multidimensional data are flattened for PrivShape.} \cite{mao2024privshape} and RetraSyn with population division (RetraSyn$^p$) \cite{hu2024real}. We use two real-world datasets (\textbf{Taxi} \cite{Taxi, yuan2010t, yuan2011driving} and \textbf{Ads} \cite{Ads, toabao}) and a synthetic dataset\footnote{As conducted by most related works, these datasets are streamed to the evaluated approaches to simulate the behavior of data streams.}. The \textbf{Synthetic} dataset is generated by a generator proposed in our previous work \cite{gu2023differential}. The introduced pattern-level PPMs \cite{palanisamy2020towards, delouee2023app} are not evaluated because they only support the sequence operator of CEP systems. Study (2): A synthetic study to analyze how our PPMs impact the data utility in a wide range of simulated practical scenarios and workloads with a sequence of synthetic datasets. Study (3): A study to evaluate the computational complexity of the proposed approaches.
\begin{figure*}
\centering
    \subfloat[Taxi with more target patterns\label{exp1a}]{\includegraphics[width=0.33\textwidth]{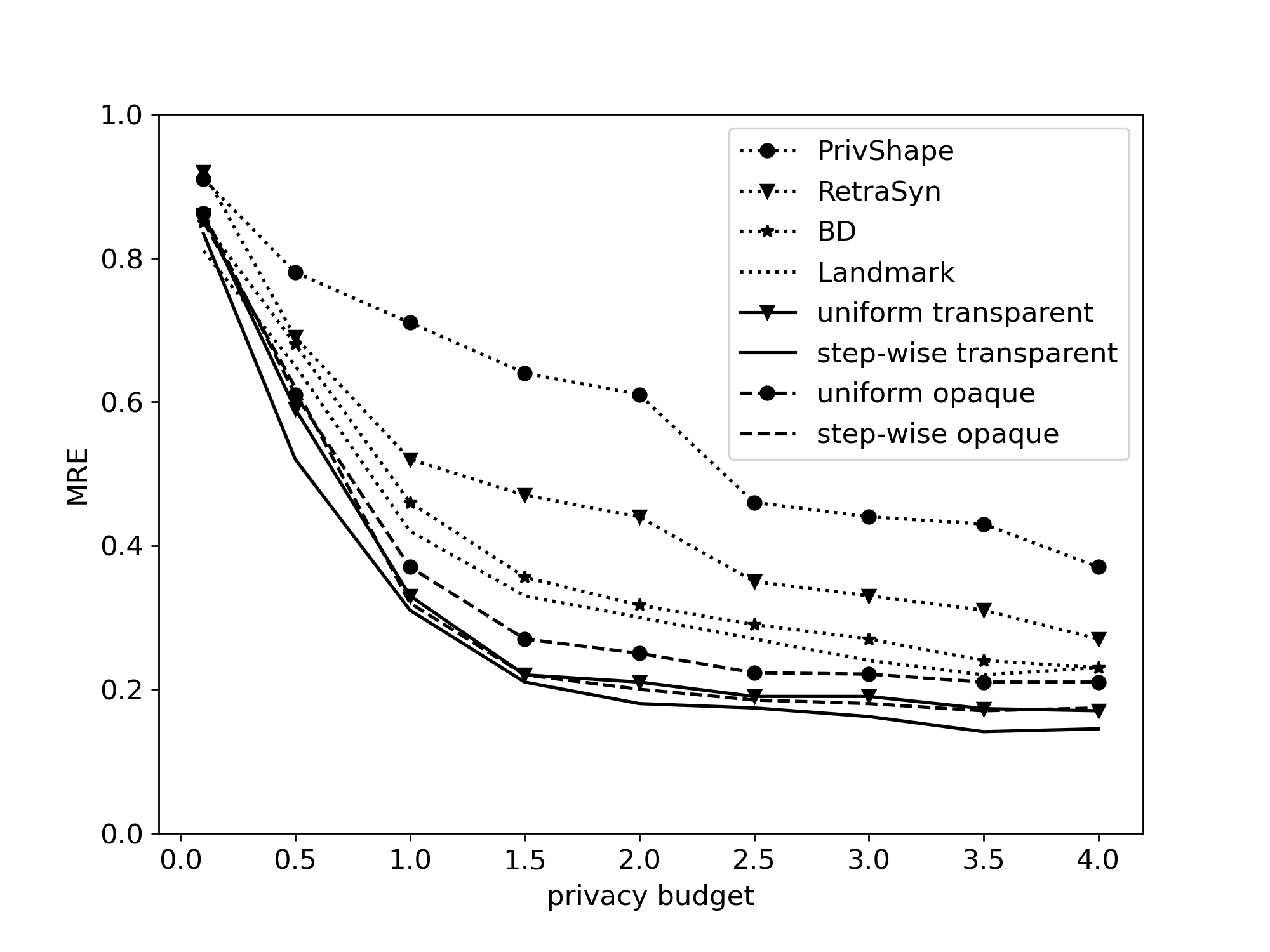}}%
    \subfloat[Advertisement\label{exp1b}]{\includegraphics[width=0.33\textwidth]{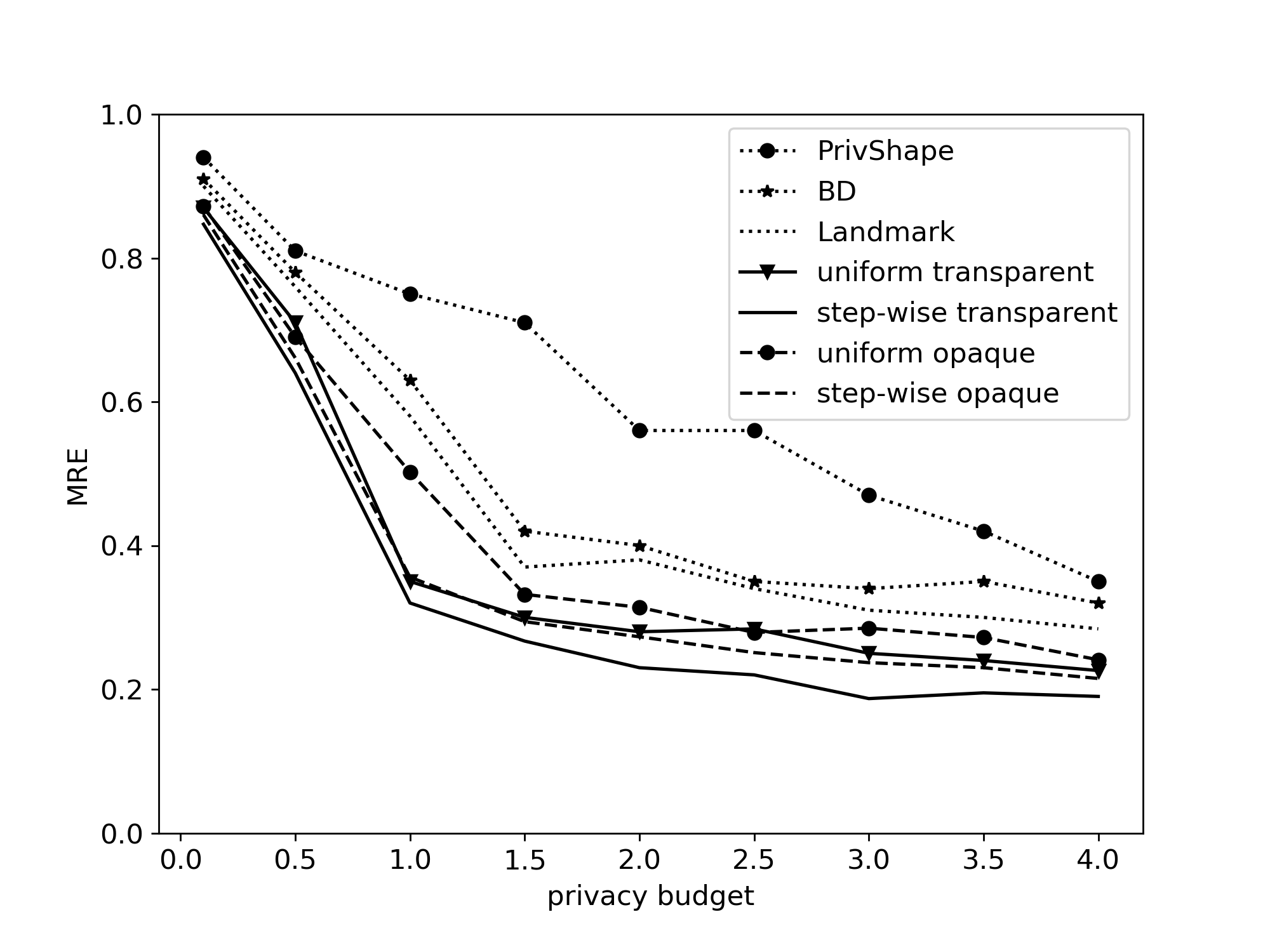}}%
    \subfloat[Synthetic datasets\label{exp1c}]{\includegraphics[width=0.33\textwidth]{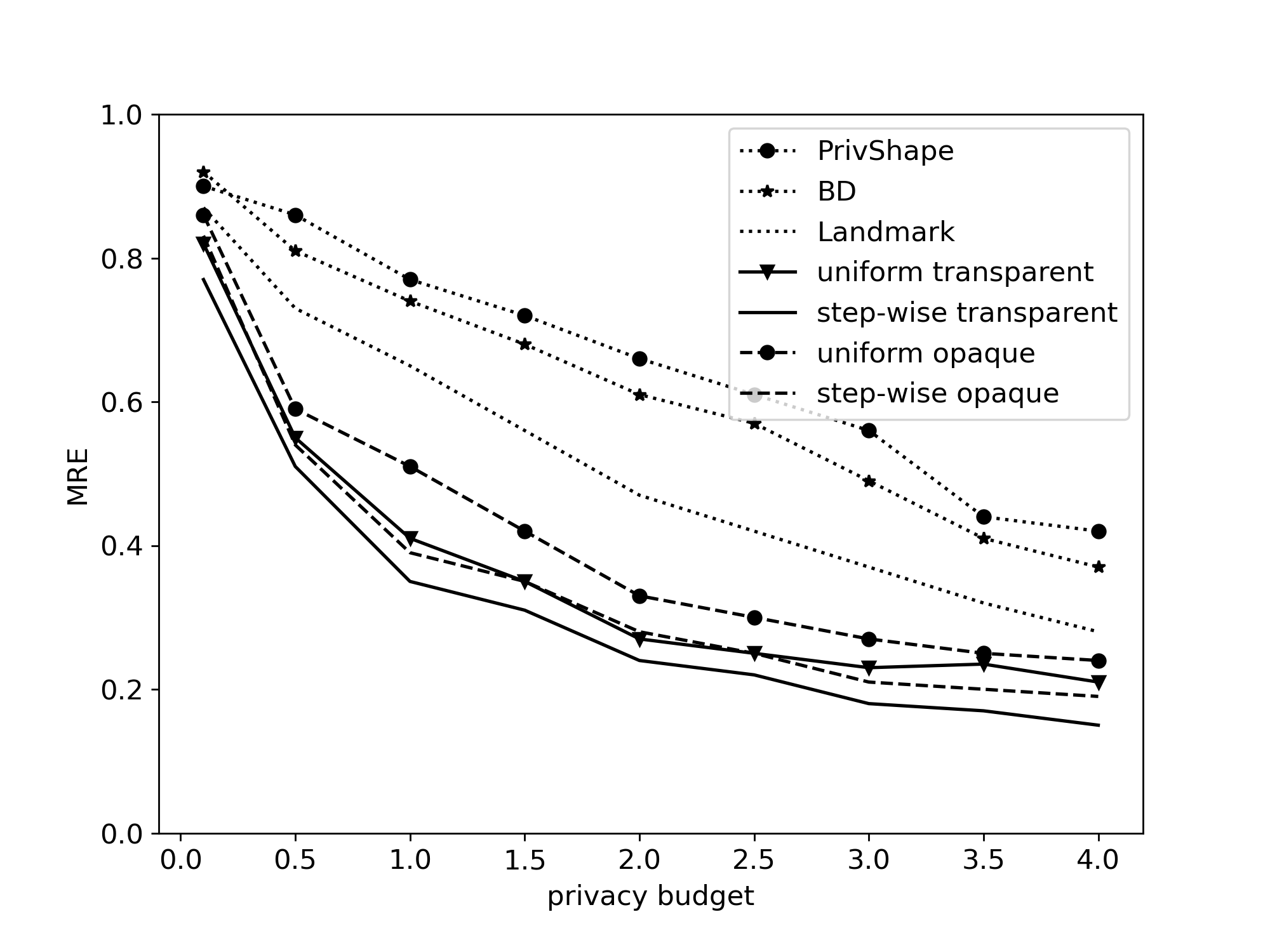}}
\caption{MRE with respect to privacy budget $\epsilon$. Higher $\epsilon$ corresponds to weaker privacy protection.}
\label{exp1}
\end{figure*}
\subsection{EXPERIMENTS}
We conduct experiments in the comparative and synthetic studies with Python on Windows 11 with an Intel 5.0GHz CPU and 32GB RAM. The study of computational complexity is conducted on a Raspberry Pi 4. The Taxi dataset contains GPS records of 10,357 taxis within Beijing (GPS locations were collected every 623 meters). Taxi has a relatively low complexity that enables an evaluation with varying ratios of private and target patterns. Ads consists of 26 million records of users of an online store, including the clicks of users on advertisements and their profiles, e.g., their IDs and genders. It contains more than 20 dimensions of data, which enables an investigation for relatively more complex use cases. Both Taxi and Ads are of sufficient quality and are used by multiple related works, including Landmark and RetraSyn. Synthetic datasets are generated by the generator \cite{gu2023differential} for more flexible analysis.

Although the datasets comprise real-world data, we need to complement them with semantically meaningful patterns to evaluate pattern-level PPMs. These patterns are related in the Taxi dataset to areas in Beijing that are determined as private areas, e.g., homes of passengers, and target areas. The Taxi service monitors any taxi that enters target areas while attempting to hide the proximity to private areas. The Beijing area is partitioned into 10,000 squared blocks of equal sizes with initially 40\% as private areas and 50\% of as target areas. Twenty-five percent of the private areas are then selected as additional target areas to ensure a sufficient overlap of private and target patterns. We conduct another experiment by exchanging target and private areas to evaluate a scenario with more private patterns than target patterns. For the Ads dataset, we construct more complex patterns to enhance the generalization of our evaluation, including four types of private patterns that combine users IDs with (1) the values of goods of interests, (2) genders, (3) resident cities, or (4) clicked advertisements, and three types of target patterns that for each user combine (1) its resident city with the goods of interests, (2) gender with clicked advertisements, and (3) age with the viewed goods. For the Synthetic dataset, patterns are constructed by the generator~\cite{gu2023differential}. By default, each pattern is constructed with six events, based on related studies and field experience, e.g., \cite{palanisamy2018preserving}. Each target pattern shares 25\% events with at least one private pattern. For the \textit{step-wise} approaches, 25\% data are randomly selected as historical data. For opaque data stream analysis, we simulate a scenario in which target patterns are confidential. 

The proposed PPMs are evaluated based on varying privacy budgets $\epsilon$ ranging from $0.1$ to $4$ for transparent and opaque data stream analysis. The performance of PPMs are evaluated by the decrease in data utility caused by employing a PPM. A smaller decrease corresponds to a better performance. The Mean Relative Error (MRE) is applied to measure the decrease:
\begin{equation}
    \mathit{MRE}_U = \frac{U_{\mathit{ord}} - U_{\mathit{PPM}}}{U_{\mathit{ord}}},
\end{equation}
where $U_{\mathit{PPM}}$ and $U_{\mathit{ord}}$ are the data utilities with and without applying PPMs, measured by $F_1$-scores. The investigated related works in the comparative study, i.e., BD, Landmark, PrivShape, and RetraSyn, are based on different DP guarantees and different definitions of privacy budgets. For evaluation, their privacy budgets are converted to the scale of pattern-level DP. BD and PrivShape are based on w-event DP and are considered as a scenario where target patterns consist of all events contained in a sliding window, and private patterns are individual events. Landmark privacy has the same private patterns, while the target patterns are its Landmark events~\cite{katsomallos2022landmark}. RetraSyn regards the transition of data producers as private patterns \cite{hu2024real} and has the same target patterns as our PPMs.

\subsection{COMPARATIVE STUDY RESULTS}
The comparative study with BD, Landmark, PrivShape, and RetraSyn consists of four experiments with (1) the Taxi dataset with more private patterns than target patterns, (2) the Taxi dataset with more target patterns, (3) the Ads dataset, and (4) the Synthetic dataset. RetraSyn is not evaluated in Experiments 3 and 4, because it is designed for trajectory data. Figure~\ref{exp1} illustrates the results of Experiments 1, 3, and 4. The results of Experiment 2 are not illustrated, as it produces a significantly similar plot to Experiment 1. The overall evaluation agrees that our pattern-level PPMs significantly outperform other state-of-the-art solutions. 

We notice that the performance of PrivShape and RetraSyn is distinct from the other approaches. Therefore, we first collectively analyze their results and subsequently discuss the performance of BD and Landmark. Compared to their related works, PrivShape and RetraSyn are notably affected by the settings of CEP systems, as illustrated in Figure~\ref{exp1}. PrivShape extracts patterns with classification methods instead of employing a CEP system. Theoretically, CEP systems can extract all patterns with $100\%$ accuracy, while the classification methods used by PrivShape have significantly lower accuracy \cite{mao2024privshape}. This leads to an unfair comparison in terms of data utility and a relatively less stable trend in MRE. Furthermore, PrivShape aims to extract statistically significant patterns, while the patterns for CEP systems are usually of little statistical significance, which reduces the effectiveness of PrivShape. The relatively low performance of RetraSyn is due to a similar cause. RetraSyn uses a Markov chain-based probabilistic model to generate synthetic trajectories \cite{hu2024real}. However, the mobility transitions based on the Taxi dataset and a CEP system are randomized, which neutralizes the usefulness of the Markov-based model. In addition, the regions labeled as private or target patterns are approximately only $70\%$ of all regions recorded by the Taxi dataset, while RetraSyn regards all regions as target patterns. This increases the redundancy in privacy protection for RetraSyn and reduces data utility.

The results of Experiment 1 (Figure~\ref{exp1a}) illustrate that the transparent and opaque \textit{step-wise} approaches outperform BD and Landmark. The transparent \textit{step-wise} approach surpasses the performance of BD and Landmark by approximately 30\% and that of the \textit{uniform} approach by 10\%, when privacy budgets are higher than 0.5. The approaches for transparent data stream analysis perform better than those for opaque data stream analysis because knowing target patterns allows PPMs to reduce the noise added to events belonging to these patterns, and thus reduces the loss of data utility. The results of Experiment 3 (Figure~\ref{exp1b}) illustrate more significantly the advantages of pattern-level PPMs. BD and Landmark are penalized more for more complex patterns, as the detection of target patterns is more affected after applying these PPMs. Opaque data stream analysis reduces the performance of our PPMs in this experiment more than in Experiment 1, since the loss of information of target patterns in such cases has more negative impacts. However, for the \textit{step-wise} approach, the information on target patterns can be partially retained by training on historical data, leading to a smaller decrease in data utility compared to the \textit{uniform} approach. Experiment 4 evaluates the PPMs for a more complex use case. Both private and target patterns for Experiment 4 consist of six events, which are significantly more than those for other experiments. The private and target patterns for this experiment are more likely to overlap each other than in other experiments. This experiment shows for all PPMs a similar trend in Figure~\ref{exp1c} as the other experiments, i.e., a higher privacy budget leads to a lower MRE. Furthermore, the graphs imply that the performance advantage of our PPMs over BD and Landmark increases with more complex patterns, as it leads to an approximately $60\%$ lower MRE in Experiment 4. Given the results of this study, we conclude that (1) our proposed approaches outperform the related works in all cases and (2) the performance advantages of our approaches increase with the complexity of the scenarios. Non-pattern-level PPMs are not designed to provide dedicated solutions for protecting private patterns, and their privacy budgets are therefore largely squandered. When privacy can be clarified and protected at the pattern level, the application of our proposed PPMs leads to significant improvements.

\begin{figure*}
\centering
    \subfloat[MRE w.r.t. occurrence of patterns\label{exp2a}]{\includegraphics[width=0.30\textwidth]{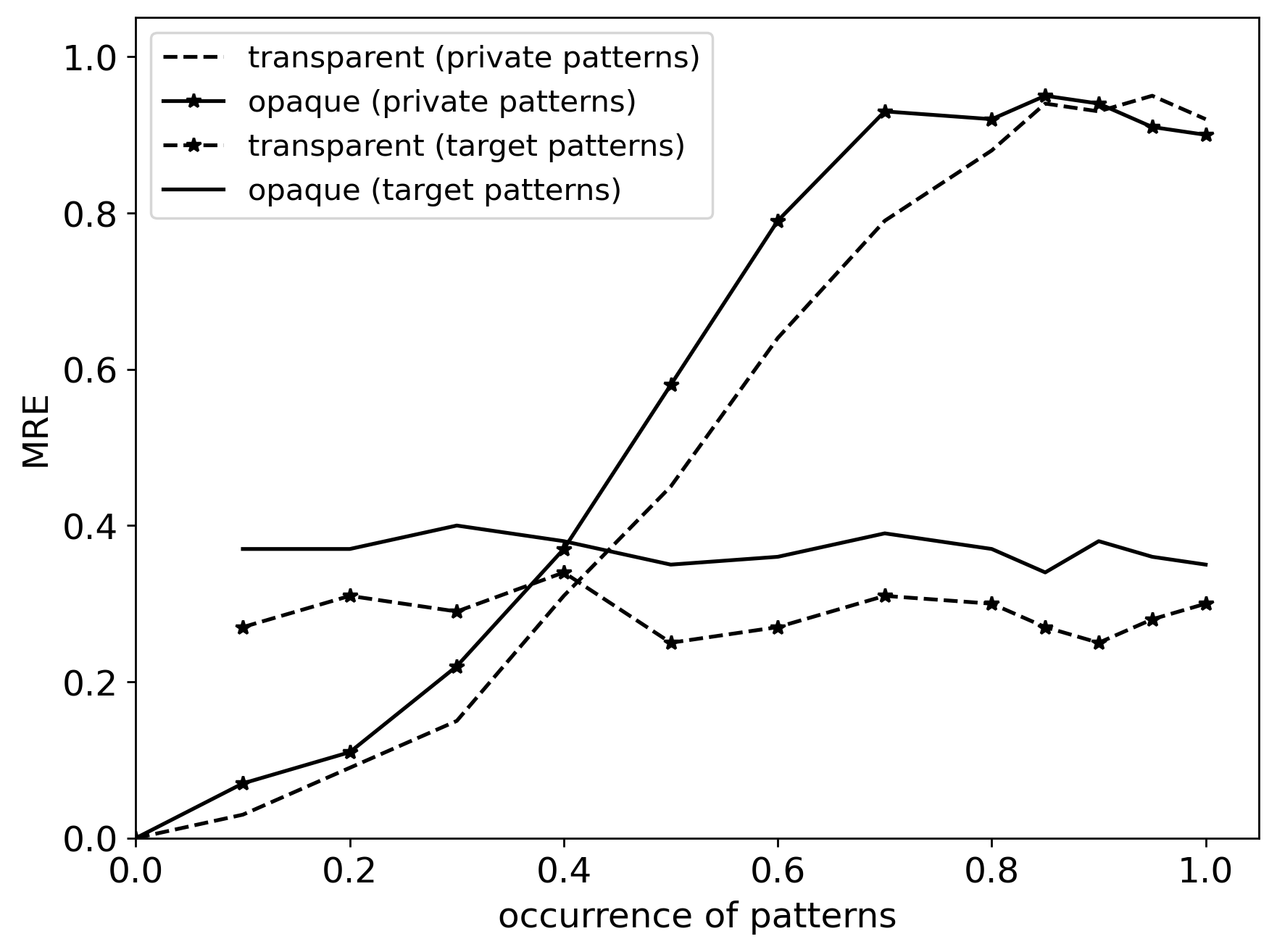}}%
    \subfloat[MRE w.r.t. complexity\label{exp2b}]{\includegraphics[width=0.30\textwidth]{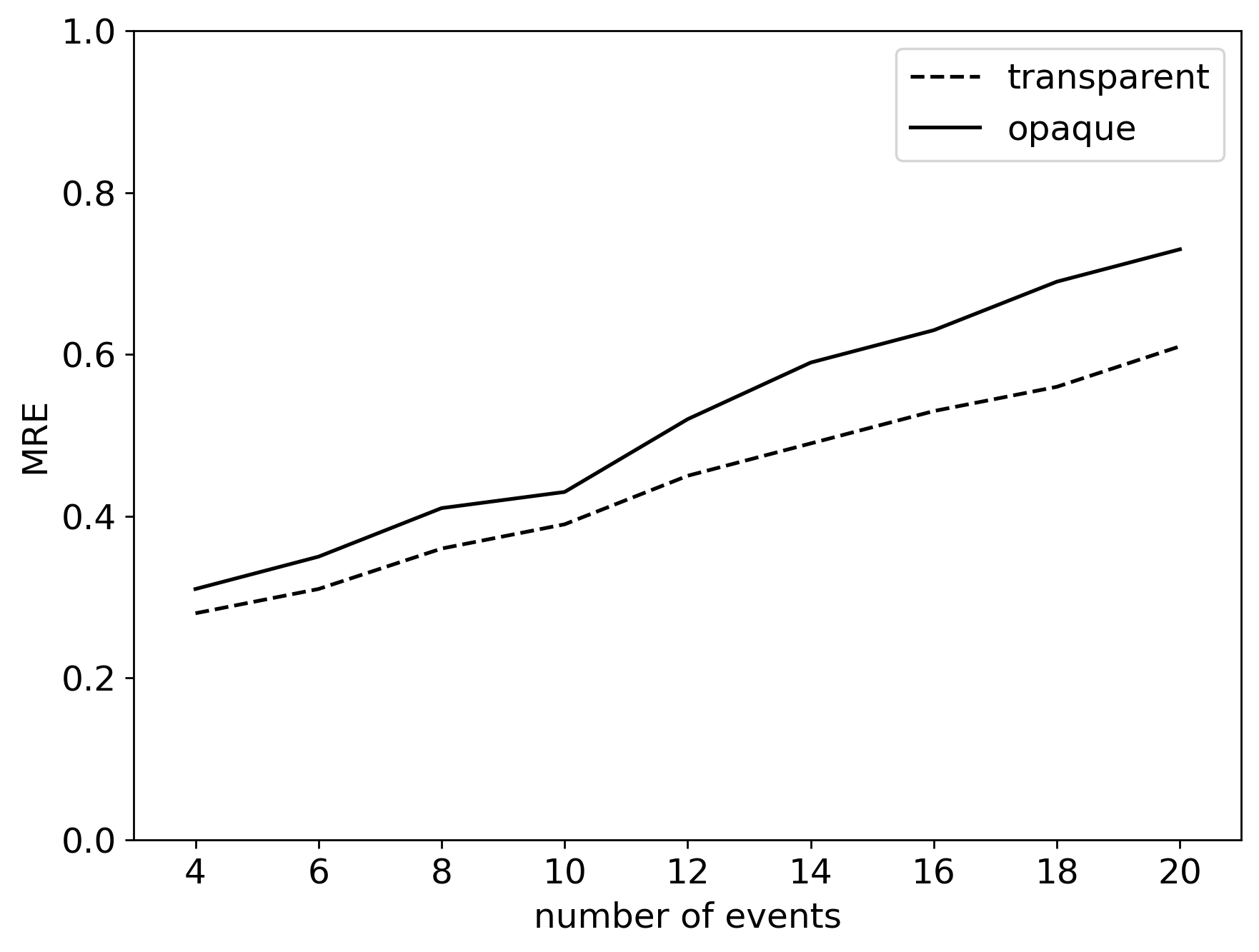}}%
    \subfloat[MRE w.r.t. overlapping rates\label{exp2c}]{\includegraphics[width=0.30\textwidth]{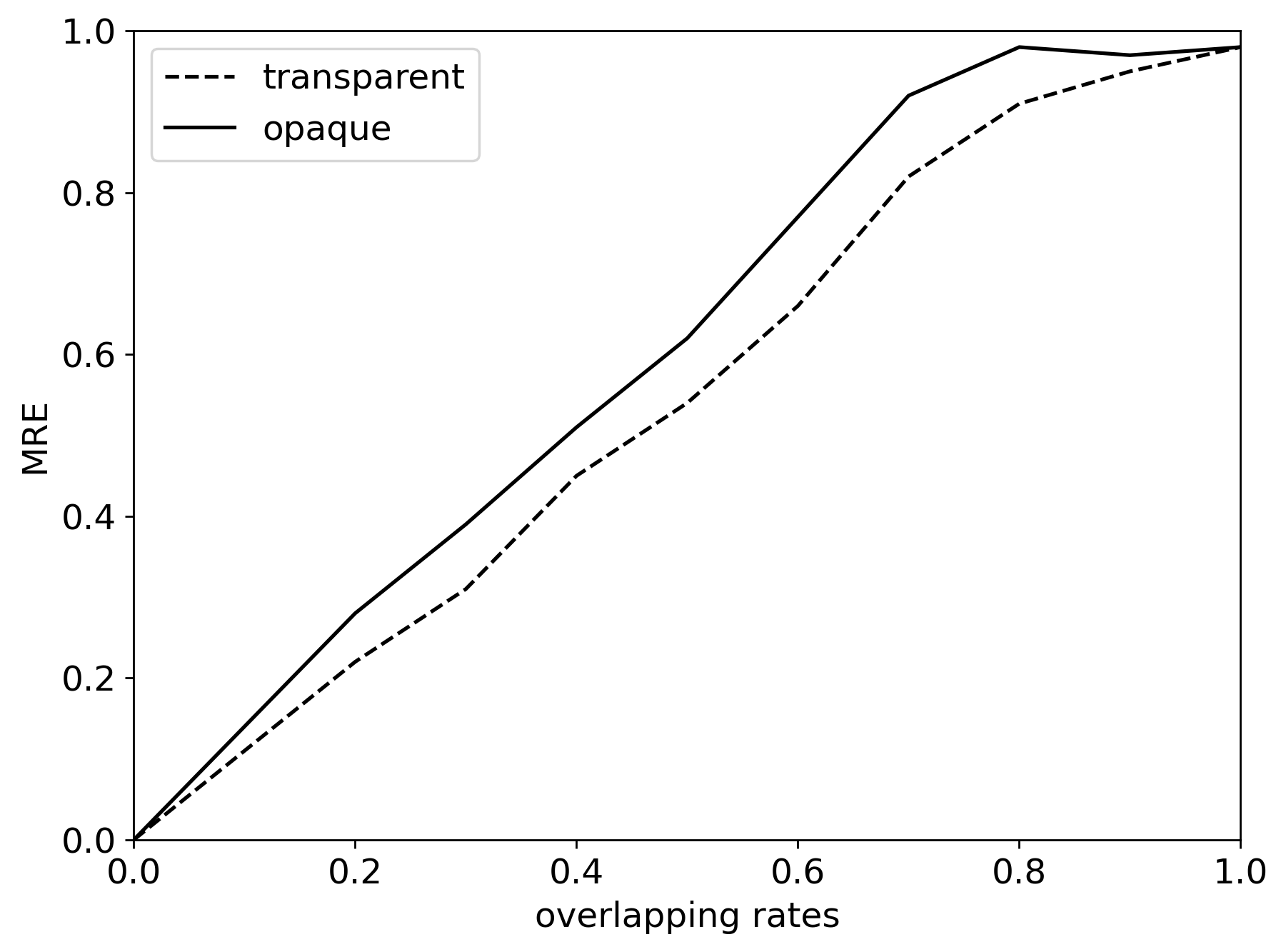}}
\caption{MREs in diverse scenarios. The privacy budget $\epsilon$ is fixed as $1$.}
\label{exp2}
\end{figure*}

\subsection{SYNTHETIC STUDY RESULTS}
The comparative study shows that our PPMs outperform the selected works for the scenarios that have been widely used for evaluation in related studies. Since many use cases might still diverge from them, we further investigate a wider range of workloads based on a sequence of synthetic datasets (each comprising $150,000$ events)~\cite{gu2023differential}. The study is conducted in three dimensions related to private and target patterns, i.e., the occurrence rates, the complexity, and the overlapping rates. The occurrence rate is the percentage of data belonging to a given pattern type w.r.t. the total amount of data. The complexity is the number of events contained in each pattern. The overlapping rates are the percentage of shared events between a target pattern and at least one private pattern. While one of the three characteristics changes, the others are fixed to default values. The default occurrence rates for private patterns are 40\%, while 50\% for target patterns, because lower occurrence rates lead to unnecessarily more experiments to collect sufficient results. The occurrence rates of private patterns are lower than those of target patterns based on field experience. The complexity is by default six events to ensure that private and target patterns share sufficient numbers of common events. In other words, the private and target patterns cannot be independent from each other. The evaluations would otherwise become meaningless. The default overlapping rates are 25\%. 

Figure~\ref{exp2a} illustrates that the MRE increases as the occurrence of privacy patterns increases, and when the occurrence of target patterns increases, the MRE is approximately constant. Higher occurrence of private patterns leads to overlapping privacy protections and redundant noises added to the same events, leading to larger MREs that eventually reach the maximum 1. However, when the occurrence of target patterns increases, the privacy protections are not affected, and hence the decrease in data utility is insignificant. The increasing complexity leads to an increasing MRE as illustrated in Figure~\ref{exp2b}. For more complex patterns, it is more challenging to maintain the completeness of target patterns from the noised data. Figure~\ref{exp2c} indicates that when the overlapping rate between private and target patterns grows, the data utility is degraded. Higher overlapping rates result in higher difficulty in maintaining the data utility of target patterns, as they are more impaired by the privacy protection employed on the shared events between private and target patterns. 

\subsection{COMPUTATIONAL COMPLEXITY}
We investigate the computational complexity of the \textit{step-wise} approach for opaque data stream analysis, as it puts the heaviest workload on data producers and consumers devices. 25\% of each dataset is randomly selected as the historical dataset required by \textit{step-wise}. We conduct the evaluation by measuring (1) the delay introduced on data producers devices (2) and the computational time on data consumers devices w.r.t. the number of data producers. 

We assume that the computing power of data producers and consumers devices, e.g., smartphones and servers, exceeds that of a Raspberry Pi. Therefore, a satisfactory computational complexity on the Raspberry Pi must be practical for data producers and consumers. To measure latency, we execute the \textit{step-wise} approach with a streamed Ads dataset in 20 events per second and measure the processing time for each event. The \textit{step-wise} approach introduces for each event an average 19ms delay (standard deviation 3.7ms, min 11ms, and max 33ms). Since it is not a performance-optimized Python implementation, we conclude that data producers devices can execute the PPM without introducing latency issues for typical applications. To measure the computational time, we evaluate the worst case in which one data consumer serves multiple distinct data producers. In detail, the computation on data consumers devices is mainly introduced in the setup phase when calculating data utilities for the \textit{step-wise} approach. In the worst case, each data producer requires a new data utility calculation based on its unique data stream schema and target patterns, while in practice, multiple data producers usually have the same data schema and target patterns so that the calculated data utilities can be cached, which significantly reduces the computational time. The evaluation shows that the computational time increases by 0.664 seconds per new producer, following a linear relationship. The amount of computational time and its linearly increasing behavior indicate a satisfactory computational complexity of our PPM in the worst case.

\section{CONCLUSION}\label{con}
State-of-the-art PPMs in CEP systems often result in redundant privacy safeguards, compromising the data utility. To reduce this redundancy, we introduced in this article novel pattern-level solutions that lead to superior performance w.r.t. data utility under equally strong privacy protection as existing solutions. The core idea is to (1) guarantee DP for private patterns and (2) leverage knowledge about private and target patterns to fine-tune the distribution of privacy budgets to maximize data utility. In transparent data stream analysis, private and target patterns are available to craft a PPM, while in opaque data stream analysis, data consumers do not reveal their target patterns. This challenge is addressed by giving data consumers the incentive of higher data utility when contributing to the implementation of a tailored PPM. All the assistance that data consumers need to provide is observation spans of the target patterns and evaluations of data utility. With this assistance, we have designed two PPMs for opaque data stream analysis that achieve pattern-level DP and are superior to existing solutions w.r.t. data utility. The core contributions of this work are: (1) A new solid theoretical foundation for privacy protections at the pattern level. (2) Four novel PPMs for transparent and opaque data stream analysis that provide pattern-level DP and superior data utility compared to existing works. (3) Experiments for a comparative study with three datasets demonstrating the improvements of the pattern-level PPMs compared to existing works. (4) A simulation-based study to systematically evaluate the performance and behaviors of our PPMs for a wide range of potential scenarios.
 
This work is based on the common assumption for most related studies that all private information is well-defined by data producers. However, in future work, we see the potential to modify this assumption so that data producers with little privacy expertise can be automatically assisted to identify and define their privacy requirements.

\bibliographystyle{IEEEtran}
\bibliography{sample-base}

\vfill\pagebreak

\end{document}